\documentclass[10pt, conference, compsocconf]{IEEEtran}

\ifCLASSINFOpdf
\else
   \usepackage[dvips]{graphicx}
   \DeclareGraphicsExtensions{.eps}
\fi

\usepackage{times,amsmath,amsfonts,epsfig}
\usepackage{graphics}
\usepackage{graphicx,epstopdf}
\usepackage{wrapfig}
\usepackage{subfigure}
\usepackage{multirow}
\newtheorem{thm}{Theorem}
\newtheorem{prop}[thm]{Proposition}
\usepackage{algorithm}
\usepackage{algorithmic}
\usepackage{framed}
\usepackage{array}
\usepackage{hyphenat}
\begin{document}

\title{Secure Computation of Top-K Eigenvectors for Shared Matrices in the Cloud}

\author{\IEEEauthorblockN{James Powers, Keke Chen}
\IEEEauthorblockA{Data Intensive Analysis and Computing Lab (DIAC)\\
Ohio Center of Excellence in Knowledge Enabled Computing (Kno.e.sis)\\
Department of Computer Science and Engineering, Wright State University\\
\{powers.4,keke.chen\}@wright.edu}
}

\maketitle

\begin{abstract}
With the development of sensor network, mobile computing, and web applications, data are now collected from many distributed sources to form big datasets. Such datasets can be hosted in the cloud to achieve economical processing. However, these data might be highly sensitive requiring secure storage and processing. We envision a cloud-based data storage and processing framework that enables users to economically and securely share and handle big datasets. Under this framework, we study the matrix-based data mining algorithms with a focus on the secure top-k eigenvector algorithm. Our approach uses an iterative processing model in which the authorized user interacts with the cloud to achieve the result. In this process, both the source matrix and the intermediate results keep confidential and the client-side incurs low costs. The security of this approach is guaranteed by using Paillier Encryption and a random perturbation technique. We carefully analyze its security under a cloud-specific threat model. Our experimental results show that the proposed method is scalable to big matrices while requiring low client-side costs.
\end{abstract}


%
\IEEEpeerreviewmaketitle

\section{Introduction}
With the development and wide deployment of web services, mobile applications, and sensor networks, data are now collected from many distributed sources to form big datasets. For example, users of mobile devices are becoming ``citizen sensors'' \cite{goodchild07} to generate information about the world via tools like Twitter and location-based services. Such datasets have become a valuable asset to the data owner. 
This paradigm raises a number of challenges for data storage, sharing, and analysis.
\begin{itemize}
\item Most data collectors, such as mobile devices and sensors, have very limited resources to store and process data. Thus, the data have to be sent to servers or clouds. 
\item The collected data may be highly sensitive. So it becomes necessary to transfer, store, and process them securely. 
\item Users might be authorized by the data owner to use the data. Processing and analyzing the secured big data will require a huge amount of computing resources to be allocated on demand, where cloud computing is the ideal platform.
\end{itemize}

With these problems in mind, we envision a cloud-based data storage and processing framework that enables users to economically and securely share and handle big matrices collected from distributed sources. The data owners and data consumers do not need to own powerful computational infrastructures - they need only a PC to interact with the cloud infrastructure to economically finish the data intensive computing. 

The key problem of this framework is the security of data. The data owner loses the control of the data once the data are exported to the cloud. For many reasons, such as a compromised cloud infrastructure or insider attacks conducted by the provider's employees, the cloud provider may not be considered a trusted party. Storing data securely is trivial, but processing data securely in the cloud is very challenging because the underlying infrastructure is owned by the untrusted cloud provider. Our framework aims to enable the trusted parties to securely conduct computation on encrypted data on top of untrusted cloud infrastructures. 

The most relevant approach is fully homomorphic encryption \cite{gentry09,naehrig11}, which, however, is too expensive to be practical. As users expect to use the cloud to minimize the client-side costs, in general, approaches, which demand high costs on communication and the client side, cannot be moved to the cloud setting (Section \ref{sec:related} for more details).   

\textbf{Scope and Contributions.} Securely sharing, analyzing, and mining datasets in the cloud is a rather broad issue. In this paper, we study a specific problem with the proposed framework: matrix-based data mining, and more specifically, we focus on the secure and efficient iterative methods for finding the approximate top-k eigenvectors of a secured large matrix in the cloud. 

Eigendecomposition \cite{saad11} has broad applications in many important areas including information retrieval \cite{deerwester90} and PageRank \cite{brin98} for ranking web search results. The common algorithm to find all eigenvectors will cost $O(n^3)$ in time complexity, which is prohibitively expensive for large dimensionality $n$. When the matrix is large, the iterative methods (i.e., power-iteration based methods \cite{arnoldi51,cullum85}) are used to find the approximate top-k eigenvectors instead. The most expensive step of the iterative methods is the matrix-vector multiplication.    


In the proposed framework, we use the \emph{partially homomorphic} Paillier encryption system \cite{paillier99} to encrypt the values which enables the processing of encrypted data in the cloud. The data owner manages the keys and owns the data in the cloud. At the data owner's request, the data collectors encrypt vectors (e.g., describing the interactions in social network or preferences over books) and submit them to the cloud to form the big matrix. The data owner or the authorized data consumer, who has only limited computing resources, then interacts with the cloud to conduct an iterative computation to find the approximate top-k eigenvectors. 

The core component of this framework is the secure matrix-vector multiplication between the encrypted matrix stored in the cloud and the vector provided by the client. Paillier encryption is used to enable computations over encrypted data in the cloud, which is much more efficient than the existing fully homomorphic encryption methods \cite{gentry09}. However, it can only provide homomorphic addition, handicapped for homomorphic matrix computation. We develop an efficient perturbation-based protocol to overcome this problem and ensure no information is leaked in the computation. It also guarantees low costs  in the trusted client side (the data owner, data collectors, and the authorized data consumers) to fully take advantage of cloud computing. We also develop a MapReduce program to process the secure matrix-vector multiplication to fully exploit the parallelism and scalability enabled by the cloud platform.


An extensive evaluation has been conducted on the proposed method. The result shows that the cloud-side parallel processing on encrypted data is efficient and scalable, and the client-side costs are quite acceptable. 

The remaining part of the paper is organized as follows. Section \ref{sec:pre} gives background knowledge about the proposed approach, including a brief description on the Paillier encryption system, the power iteration methods, and the MapReduce programming model. Section \ref{sec:approach} describes the design of the approach and also analyzes its correctness, security, and costs. Section \ref{sec:exp} presents the results of experimental evaluation with a focus on the storage and computing costs in both the cloud and client sides. Section \ref{sec:related} shows some related work on secure computation in the cloud.

\section{Preliminaries} \label{sec:pre}

In the following, we will briefly describe our notation and background knowledge about eigenvalue decomposition, Paillier encryption, and MapReduce programming. 

\textbf{Notation}
For clear presentation, we will use Greek characters to represent scalars, lower case letters for vectors, indexed lower case letters for the elements in the vector, and capital letters for matrices or submatrices. In particular, $\mathbb{Z}_q^n$ represents the group of $n$-dimensional vectors of integers with moduluo $q$.
We use $\{b_i\}, i=1..k$ to denote a set of vectors (or values). 


\textbf{Power Iteration for Eigenvalue Decomposition.} Finding eigenvalues and eigenvectors of a square matrix has many important applications in science and engineering domains. In particular, eigenvalue decomposition has been an important tool in data analysis. For example, Principal Component Analysis (PCA) depends on eigenvalue decomposition \cite{jolliffe86}. In information retrieval, it is used to find the relationship between words and documents on a text corpus \cite{deerwester90}. Google's patented PageRank technique\cite{brin98} is also related to eigenvalue decomposition. 

The matrices from these applications are normally very large. Thus, the direct computation method that costs $O(n^3)$ is not a viable option. Instead, for large matrices, power iteration based methods, such as the Arnoldi \cite{arnoldi51} and Lanczos methods \cite{cullum85}, are used for finding the top-k eigenvectors. 
Assume $A$, $A\in \mathbb{R}^{n\times n}$, is a $n\times n$ real matrix and $b_0$ is a random $n$-dimensional vector. We sketch these methods in Algorithm \ref{alg:power-it}.
\begin{algorithm}[htb]
\caption{Framework of Power Iteration Methods}\label{alg:power-it}
\begin{algorithmic}[1]
\STATE $b_0 \leftarrow$ random $n$-D vector;  
\FOR{i $\leftarrow$ 1 to $k$}
\STATE $b_i\leftarrow Ab_{i-1}/|| Ab_{i-1}||$;
\STATE other operations of cost $O(n)$ specific to the Arnoldi or Lanczos methods.  
\ENDFOR
\STATE Post-processing with a cost $O(n)$ to generate eigenvectors. 
\end{algorithmic}
\end{algorithm}

Note that in this iterative framework, the most expensive operation is the line $b_i\leftarrow Ab_{i-1}/|| Ab_{i-1}||$. Other operations are only related to processing $b_i$ and some auxiliary $k\times k$ matrix. Usually, only a few eigenvectors are needed and thus $k$ is typically small (e.g., k=10). The remaining computations, specific to the Arnoldi and Lanczos methods cost only $O(kn)$. With the matrix stored in the cloud, we can let the cloud take care of the most expensive part while the client handles the remaining low-cost steps.

\textbf{Paillier Encryption.} 
Fully homomorphic encryption aims to allow addition and multiplication to be conducted on encrypted values without decryption. Paillier Encryption is a partially homomorphic encryption scheme that is much more efficient than the fully homomorphic ones \cite{gentry09}, but only preserves homomorphic addition, i.e. 
\begin{equation}
E(x) + E(y) = E(x+y).
\end{equation}
Multiplication cannot be implemented on top of $E(x)$ and $E(y)$ in Paillier encryption. However, with one operand not encrypted, say $y$, multiplication can still be implemented as
\begin{equation}
E(xy) = E(x) ~\text{mod\_power} ~ y,
\end{equation}
where \emph{mod\_power} means the modulo power operation \cite{katz07}. For simplicity of presentation, we use $(E(x))^y$ to represent \\$E(x)\ \text{mod\_power}\ y$. It has been proven that Paillier Encryption has strong security guarantee, satisfying the definition of semantic security.

Therefore, to implement more complicated operations with Paillier encryption, one has to expose one of the operands which becomes the security hole. Finding a way to limit the exposure and maintain data privacy will be one of the challenging tasks in the application.


\section{Secure Top-k Eigenvector Computation in the Cloud}\label{sec:approach}
In this section, we will describe the major components in our approach. We begin with the general computational framework and describe the roles of the client and cloud components. Next, we present the security/threat model and discuss our assumptions and potential attack points. We follow with the secure power iteration algorithm that is built up with Paillier encryption and random perturbation. Then, we briefly discuss the cloud-side MapReduce algorithm. Finally, we formally analyze the security and cost of the proposed approach. We show that our approach can minimize the cost of the client-side computation while providing excellent performance for processing big matrices with strong security guarantee.

\subsection{Computational Framework}

Our framework involves four parties: cloud, data owner, data collectors, and authorized data users. It reflects the highly distributed nature in data intensive computing where collecting data, storing data, and processing data might be handled by different distributed parties. 

Figure \ref{fig:framework} illustrates the relationship and interaction between these parties. The non-cloud parties: the data owner, the data collectors/contributors, and the authorized users are trusted. The data owner controls all rights to the data, distributes public keys, and asks the data collectors/contributors to upload the collected data which are encrypted by a public key. The data owner can process the pooled data by him/herself or authorize other users to use the encrypted data. Each data collector may contribute a small part of the data, e.g., a row of the matrix. Practical examples may include the interactions with other users or the recommendations on items during a period. The data owner or the authorized users can interact with the cloud to conduct matrix analysis tasks. Note that encrypted data will have a size much larger than the original one. The authorized users cannot afford to download the encrypted big matrix and conduct computations locally. Instead, they want to fully utilize the benefits of cloud computing and minimize the client-side cost. 
\begin{figure}[tbh]
\centering
\begin{minipage}{\linewidth}
\centering
\includegraphics[width=\linewidth]{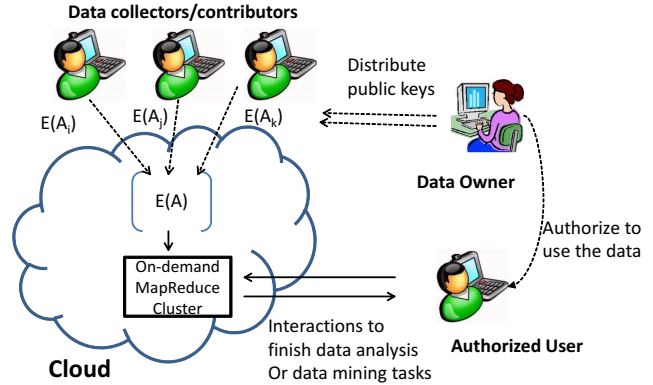}
\caption{Framework for conducting matrix mining with the cloud.}\label{fig:framework}
\end{minipage}
\end{figure}

\subsection{Threat Model}

\textbf{Assumptions.} Our security analysis is built on the important features of the discussed architecture. Under this setting, we believe the following assumptions are appropriate. 
\begin{itemize}
\item The cloud provider is not trustable and may silently observe the data and the computation to find useful information, e.g., the eigenvectors.
\item Only the data owner and the authorized users can use the proprietary matrix. The data owner authorizes some trusted users to use the data. They will not intentionally breach the confidentiality. We consider insider attacks to be orthogonal to our research.
\item The client-side system and the communication channels are properly secured and no part of the confidential matrix or the computation results can be leaked. 
\item Adversaries can see the secured matrix and the submitted plaintext vectors for secure matrix-vector computation, but nothing else.
\end{itemize}
These assumptions can be maintained and reinforced by applying appropriate security policies. 
 
\textbf{Protected Assets.} Data confidentiality is the central issue in our approach. While the integrity of data and computation is also an important issue, it is orthogonal to our study. Due to the space limitation, we will not cover data and computation integrity techniques \cite{wang11}, which are typically used to prevent adversaries from actively tampering with the data or computation. Therefore we can assume the ``honest but curious'' adversary model. 

\textbf{Attacker's goal.} The attacker is interested in recovering (or estimating) the matrix elements and the computing result, i.e., the eigenvectors and eigenvalues. 

\textbf{Security Definition.} The protected matrix is indistinguishable to chosen-plaintext attack (IND-CPA). The submitted plaintext vectors in the computation are no different from random samples uniformly drawn from a sufficiently large vector space (with $O(2^K)$ elements in the set, if $K$ is the number of bits for encoding). The protocol will not reveal any additional information.  

\subsection{Data Collection}\label{sec:data}
\textbf{Representation of Data.} In order to use the Paillier encryption system all data need to be converted to \emph{non-negative big integers} \cite{paillier99}. However, in practice, we process matrices in the real value domain. A valid method to preserve the desired precision, e.g., $d$ decimal places, is to multiply the original value by $10^d$ to scale up the values and drop the remaining decimal places. In addition, we also need to shift the values so that the domain is positive. This process is clearly reversible so that the results can be correctly recovered. With a typical key length of 1024, we have enough digits to preserve the precision. 



\textbf{Submitting data.} To prepare for collecting data, the data owner will generate one $n$-dimensional random vector $b_0$, $b_0\in \mathbb{Z}_q^n$. $b_0$ is then encrypted with the Paillier public key:\\ $E(b_0)=(E(b_{01})\ldots E(b_{0n}))$, which is then distributed to the data collectors. 

The data collectors submit their row(s) of the matrix, denoted as $A_i$ for the collector $i$, in the encrypted form to the cloud storage. In addition, they also calculate the result of $E(A_ib_0)$ with the following homomorphic method and submit it to the data owner. Assume $a$ is one row of $A_i$
\begin{equation} \label{eq:matrix-vector1}
E(ab_0) = \sum_{k=1}^n E(a_k b_{0k})
= \sum_{k=1}^n (E(b_{0k}))^{a_k},
\end{equation}
Note that the number of elements in $E(A_ib_0)$ is the same as the number of rows to be submitted to the cloud by the collector, which is typically one. Finally, the data owner collects all $E(A_ib_0)$ and decrypts them to find $Ab_0$.

\subsection{Secure Power Iteration Protocol}

To protect the plaintext vectors submitted to the cloud in power iteration, the authorized data user must perform a few steps to prepare for the perturbation approach. Then, the client side collaborates with the cloud side to finish the secure matrix-vector multiplication in the iterations. 

\textbf{Preparing the Perturbation Pool.} The authorized data user will receive $E(b_0)$, $E(Ab_0)$, and the decryption key from the data owner, and then select $m$ $n$-dimensional random vectors, where $m$ is small, say $m=5$, and send them to the cloud. These random vectors will be used to perturb and protect the vectors $\{b_i\}$ in each iteration. Let's denote them as the seed random vectors $\{s_i\}$, for i=1..m, $s_i \in \mathbb{Z}_q^n$. 

For each random vector $s_i$, a secure computation of $As_i$ is performed in the cloud as follows. With the homomorphic properties of Pailier encryption, for the $j$-th element of the resultant vector $(As_i)_j$, we have
\begin{equation} \label{eq:matrix-vector2}
E((As_i)_j) = E(\sum_{k=1}^n A_{jk}s_{ik}) = \sum_{k=1}^n (E(A_{jk}))^{s_{ik}},
\end{equation}
where $s_{ik}$ is the $k$-th element of the vector $s_i$ and $A_{jk}$ is the element $(j, k)$ of the matrix $A$. 
Note this secure matrix-vector computation is also used later in each iteration. The result $E(As_i)$ is sent back to the client side and decrypted for later processing. 
After the preparation stage, the authorized user holds the random vectors $S =\{s_i\}$ and the resultant vectors $A_S=\{As_i\}$, for i=1..m.

\textbf{Iteration}. The iteration stage starts with the random vector $b_0$, then applies the formula $b_{k+1} = Ab_k/||Ab_k||$ and other low-cost steps as described in the Arnoldi and Lanczos methods. It is important to protect $b_i$ in each iteration - otherwise, the eigenvectors are revealed. 
We apply the following method to protect the privacy of computation, the security of which will be analyzed in detail later.

To calculate $E(Ab_i)$ from $E(A)$ and $b_i$ with the Paillier homomorphic operations, $b_i$ cannot be encrypted. We design a perturbation method to protect $b_i$ before sending it to the cloud. The basic idea is to use a random vector $r_i$ and send 
\begin{equation} \label{eq:hidden-vector}
\bar{b}_i=b_i+r_i \bmod q
\end{equation}
to the cloud instead, where $q$ is a big random prime number so that $q$ is large enough to contain all the values in the application domain. We design $r_i$ with the seed random vectors generated during the preparation stage:
\begin{equation}
r_i = \sum_{l=1}^{m} \alpha_{il} s_l + \sum_{j=0}^{i-1}\beta_{ij} b_j\bmod q,
\end{equation}
for $i=1..k$, where $\alpha_{il}$ and $\beta_{ij}$ are randomly drawn from $\mathbb{Z}_q$. The purpose of including $b_j, j=0..i-1$ in perturbation is to provide better security, which will be discussed later. As the results, $\{As_k\}$ and $\{Ab_j, j<i\}$, have been computed in the preparation stage and the previous steps, $Ar_i$ can be conveniently calculated by
\begin{equation}
Ar_i = \sum_{k=1}^{m} \alpha_{ik} As_k + \sum_{j=0}^{i-1}\beta_{jk}Ab_{j_b} \bmod q, 
\end{equation}
with only the client-side vector operations (a $O(n)$ cost). 

Then, the cloud side will take $\bar{b}_i= b_i+r_i \bmod q$, apply formula \ref{eq:matrix-vector2} to calculate $E(A\bar{b}_i)$, and send the result back. The client side decrypts $E(A\bar{b}_i)$ to get $A\bar{b}_i$. With $Ar_i$ known, we have $Ab_i = A\bar{b}_i - Ar_i  \bmod q $. Once we have $Ab_i$, it is easy to calculate $b_{i+1} = Ab_i/||Ab_i||$ for the next iteration. Additional operations (as described in the Arnoldi and Lanczos methods) have a cost of $O(kn)$, which can be conveniently performed with only client-side operations.

\subsection{Cloud-side MapReduce Computation}

Data encrypted with Paillier encryption are significantly larger than the unencrypted values. With a 1024-bit key, a 64-bit double-type original value becomes a 2048-bit encrypted one, a 32-time increase. This cost cannot be avoided by using any encryption schemes based on the assumption of Diffie-Hellman or large-integer factorization \cite{katz07}. This literally turns a common-size problem to a ``big data'' problem, which requires us to exploit the parallel processing power in the cloud.   

With this problem in mind, we designed the MapReduce version of homomorphic matrix-vector multiplication. The client passes the perturbed vector $\bar{b}_i$ as a parameter to the MapReduce program and the cloud computes and returns $A\bar{b}_i$. 
Below we describe the MapReduce formulation of the cloud-side computation of $E(A\bar{b}_i)$.

\newcommand{\map}{\ensuremath{\mbox{\bf map}}}
\newcommand{\reduce}{\ensuremath{\mbox{\bf reduce}}}
\newcommand{\partition}{\ensuremath{\mbox{\bf partition}}}
\begin{algorithm}[htb]

\caption{The MapReduce Matrix-Vector Multiplication program on encrypted matrix}
\begin{algorithmic}[1]
\STATE \map$(E(\tilde{A}), \bar{b}_i)$
\STATE $E(\tilde{A})$: some rows of the encrypted matrix $A$ that are distributed to the specific Map; $\bar{b}_i$: the perturbed vector sent by the client.
\FOR{each row of $E(\tilde{A})$: $E(A_j)$}
\STATE Emit($\langle j, \sum_{k=1}^d (E(A_{jk}))^{\bar{b}_{ik}} \rangle$)
\ENDFOR
\end{algorithmic}
\medskip
\begin{algorithmic}[1]
\STATE \partition$(j, nr)$
\STATE $j$: the row number; $nr$: the total number of reduces.
\STATE return $\lfloor j/nr \rfloor$;
\end{algorithmic}
\medskip
\begin{algorithmic}[1]
\STATE \reduce($\langle j, v \rangle$)
\STATE $j$: the row number; $v$: the result.
\STATE Emit$(\langle j, v\rangle )$;
\end{algorithmic}
\end{algorithm}

The MapReduce program is rather straightforward. The Map function applies the secure matrix-vector multiplication formula (Eq. \ref{eq:matrix-vector2}), and emits the results indexed by the row number. The Map outputs are partitioned and sorted by row number and sent to the corresponding identity Reducer which writes the data segment to disk. Because we used the binary representation for the encrypted elements of the matrix,  we also designed special input/output format classes to handle the binary data.

\subsection{Security Analysis}
The proposed algorithm consists of three components (1) data collection, (2) the random perturbation step in the client side, and (3) the matrix-vector multiplication based on data encrypted with the Paillier encryption scheme. Component (1) and (3) are secure as long as the Paillier encryption is secure. 
Thus, the security of the approach only depends on that of component (2). 

In perturbation preparation, the curious cloud provider is able to collect the initial random seed vectors $S=(s_1,\ldots,s_m)$ in the perturbation pool and the progressively generated perturbed vectors $\{\bar{b}_i\}$. In addition, the adversary could be aware that $b_i$ may converge to the dominant eigenvector corresponding to the largest eigenvalue in $k$ iterations \cite{saad11}.

\textbf{Security of $\{b_i\}$ in a single run.} The first problem is whether the adversary can gain additional information by observing the known vectors $\{\bar{b}_i\}$ and $S$. We want to show that: 
\begin{prop}
The known $\{\bar{b}_i\}$ and $S$ do not reveal any information about $\{b_i\}$.   
\end{prop} 
\begin{proof}First, we prove the $\bar{b}_1$ case. Other cases are similar. 
Let $a_i =(\alpha_{i1},\ldots,\alpha_{ik})$. Recall that Equation \ref{eq:hidden-vector} is $\bar{b}_i =b_i+r_i= b_i + Sa_i + \sum_{j=0}^{i-1}\beta_{ij} b_j\bmod q$ for $i=1..k$. Thus, $\bar{b}_1= b_1+r_1 = b_1 + Sa_1 + \beta_{1,0}b_0 \bmod q$. With the known $\bar{b}_1$ and unknown $r_1$, it is clear that if the adversary can guess $b_1$ from any uniformly random sample drawn from $\mathbb{Z}_q^n$ with non-negligible advantage, then she/he can also distinguish $r_1$ from random vectors. 

Let $r_1$ be represented as $r_1=Sa_1 + e_1 \bmod q$, where $e_1=\beta_{1,0}b_0$ and $a_1$ is secret. If $S$, $\beta_{1,0}$ and $b_0$ are drawn uniformly at random, the problem of distinguishing $<S, r_1>$ from uniformly random samples over $\mathbb{Z}_q^{n\times m} \times  \mathbb{Z}_q^n$ is exactly the decision version of the \emph{Learning with Errors} (LWE) problem discussed by Regev \cite{regev05}. It is already known that such $r_1$ cannot be distinguished from uniformly random samples if $e_1$ is randomly drawn and $a_1$ are secret \cite{regev05}. Therefore, $b_1$ cannot be distinguished from uniformly random samples as well. The same conclusion can be extended to the cases of $i>1$ with more unknowns included. We skip the details here.

Because $r_i$ cannot be distinguished from uniformly random samples, regardless of how $b_i$ appears (as $b_i$ will look similar with sufficiently large $i$), $\{\bar{b}_i\}$ cannot be distinguished from any set of random vectors. Thus, the series $\{\bar{b}_i\}$ does not help the adversary gain any additional information about $\{b_i\}$.  
\end{proof}

\textbf{Statistical Inference Attack.} The curious cloud provider may look at multiple runs of eigendecomposition conducted by different users. As all users start with the same $b_0$, will the multiple runs provide an opportunity for statistical inference? Below we describe an inference attack that looks at the statistical property of the series and analyze the risk under this attack.

Again, we start with the simplest case $\bar{b}_1=b_1+r_1$. The statistical inference attack treats $\bar{b}_1$ and $r_1$ as random variables and tries to estimate $E[\bar{b}_1]$ and $var(\bar{b}_1)$ with random samples. Theoretically, $E[\bar{b}_1] = b_1 + E[r_1]$ and $var(\bar{b}_1)=var(r_1)$. Correspondingly, the estimate $\hat{b}_1 = E[\bar{b}_1] - E[r_1]$. In practice, this estimation has to depend on  $N$ samples of $\bar{b}_1$: $\{\bar{b}_1^{(i)}, i=1..N\}$ observed by the attacker. We want to show that:
\begin{prop}
The proposed random perturbation method is computationally secure to the statistical inference attack.
\end{prop}
\begin{proof}
We use the $\bar{b}_1$ case to prove the statement. Let's analyze the statistical property of $r_1$ to see the effectiveness of this attack. 
As shown in the earlier analysis, $r_1$ cannot be distinguished from a uniformly random sample in the domain $\mathbb{Z}_q^n$. In the integer domain $\mathbb{Z}_q$, a uniformly random variable $v$ has a mean value $E[v] =q/2$ and variance $var(v)=q^2/12$.
Thus, we can assume it is from the same distribution and each element has the properties $E[r_{1i}] = q/2$ and $var(r_{1i})=q^2/12$, for $i=1..N$. The estimate $\hat{b}_1$ of $b_1$ can be established as $\hat{b}_1 = 1/N\sum_{i=1}^N \bar{b}_1^{(i)} - E[r_1]$. However, the accuracy of this estimate is determined by the variance of $\bar{b}_1^{(i)}$, which has the same variance as $var(r_1)$. With  $N$ independent samples the variance of $\hat{b}_1$ is
\[
 var(\hat{b}_1) = var(\frac{1}{N} \sum_{i=1}^N \bar{b}_1^{(i)}) = \frac{1}{N} var(\bar{b}_1)=\frac{1}{N}var(r_1).
 \]
To gain sufficiently small variance for the estimation, the number of samples $N$ would be as large as $\Theta(q^2)$. In the case of using long integers, $q^2$ will be in the range of $2^{128}$, which can be further increased if a longer representation is used. It is impractical that the authorized users will conduct this number of computations. Therefore, the statistical inference attack is ineffective. 
\end{proof}

\subsection{Cost Analysis}
The cost of the approach can be analyzed according to the three parties we described. 

\textbf{Data Collector} needs to encrypt one vector, which costs $O(n)$ value encryptions, conduct one homomorphic dot product, which has a similar cost to vector encryption (see Section \ref{sec:exp}), and transmit the encrypted vector to the cloud, which has a cost of $O(n)$. In total, its computational and transmission costs are $O(n)$.

\textbf{Authorized User} needs to prepare the perturbation pool, which costs $O(mn)$ in transmission and also $O(mn)$ decryptions, where $m$ is small. In power iterations, it needs $O(kn)$ decryptions in total for getting the $k$ eigenvectors, and $O(kn)$ space for storing the encrypted values. As $k$ and $m$ are small, the client side costs are small - in general, a PC can handle such a workload. This is an important feature for users to fully enjoy the benefits of cloud computing.

\textbf{Cloud Side} has a $O(n^2)$ storage cost for encrypted matrix elements, as well as computation costs on MapReduce-based secure matrix-vector multiplication. Because the major cost is on the Map phase, with $p$ Map slots in the Hadoop cluster, the total cost is about $O(n^2/p)$. The storage cost is proportional to the number of encrypted values, and related to the key size. 

We will conduct extensive experiments to carefully evaluate these costs.

\section{Experiments}\label{sec:exp}
 
To show the effectiveness of the proposed research, we conducted a set of experiments to evaluate the efficiency of processing in the three involved parties: the data collectors, the authorized users, and the cloud.

\subsection{Setup}
The client machine is configured with 128 GB of RAM and four quad-core AMD processors. The cloud-side MapReduce program was tested using the Hadoop cluster at Wright State University. The cluster is configured with 16 slave nodes running Apache Hadoop version 1.0.3. Each slave node is configured with 16 GB of RAM, four quad-core AMD processors, 16 map slots, 12 reduce slots, and a 64MB HDFS block size.  The cloud-side MapReduce program was implemented with Java, and the client-side programs with C++ and the GMP library (gmplib.org).

We use a 1024-bit key in our experiments. Paillier encryption with key sizes less than 1024 is considered not secure \cite{lenstra01}. 
The experiments are conducted with simulated matrices. Because the studied problem is fundamental and general to all applications it is sufficient to use simulated data. The original matrices use double values as the elements (8 bytes per value). As discussed in section \ref{sec:data}, these matrices are converted to long integers for encryption to preserve sufficient precision. The encrypted matrices are used as input to the MapReduce jobs in the power iteration algorithm. 

\subsection{Data Collector Costs}
Each data collector in the framework will generate one vector (or a few), encrypt them, and deliver them to the cloud. In addition, it will conduct the secured vector dot-products for generating $E(A_ib_0)$. Thus, the data collector's major costs are on vector encryption, secure vector dot-products, and transmission. These costs are determined by the basic cost of Paillier encryption (i.e., the key size) and the number of dimensions.

\begin{table}[tbh]
\centering
\scriptsize
\begin{tabular}{|c|c|c|c|c|}
\hline
\multirow{2}{*}{Encoding} & Encrypt & Dot Product & Size & Compressed \\
& (sec) & (ms) & (bytes) & (bytes) \\
\hline
Text& 56 & 3.7 & \textbf{\emph{6.1M}} & \textbf{\emph{3.5M}} \\
Binary& 56 & 3.7 & 2.5M & 2.5M\\
\hline
\end{tabular}
\caption{Data Collector costs for a 10,000-dimension vector.} \label{tab:collector-costs}
\vspace{-0.5cm}
\normalsize
\end{table}
 
We use the binary encoding scheme to minimize the size of encrypted data, which also minimizes communication and cloud storage. Basically, the size of an encrypted value in binary representation has twice the size of the encryption key, e.g. a 64-bit double-type value will become 2048-bit encrypted value with a 1024-bit key. We show a simple comparison using a 1024-bit key for a 10,000-dimension vector in Table \ref{tab:collector-costs}. 

Compared to the binary representation, a naive text representation of big integers without compression will cost about 150\% more in space (and 40\% more with compression) but the computation time is about the same. 
The binary representation does not significantly compress due to the randomized nature of encrypted values. 

\subsection{Client-side Costs in Iterations}
In each iteration of the proposed algorithm, the client side will receive the encrypted vector $E(A\bar{b}_i)$ from the cloud, decrypt it, regenerate a plaintext perturbed vector, and submit it for the next iteration. We will evaluate the communication, memory, and computation costs of the corresponding steps.  

The communication costs consist of receiving the encrypted vector from the cloud and sending the plaintext perturbed vector back to the cloud. We use the number of bytes to represent these costs. According to the previous discussion, an encrypted vector of 10,000 dimensions will cost 2.56M bytes, which needs to be transmitted to the client side. In comparison, the returned plain vector of 10,000 long integers has 80K bytes. These sizes are also linearly proportional to the number of dimensions.

The computational steps include decrypting the vector and constructing the new plaintext vector. Let the perturbation pool contain 10 randomly generated vectors. Table \ref{tab:client-comp-cost} shows the cost distribution with different dimensions. As observed, decryption takes most of the time. Since decryption can be easily done in parallel with a multicore processor, this cost can be further reduced.   

\begin{table}[tbh]
\centering
\scriptsize
\begin{tabular}{|c|c|c|c|c|}
\hline
Dimension & $E(A\bar{b}_i)$ size & Decrypt $E(A\bar{b}_i)$ & Other processing \\
\hline
10000& 2.56MB& 31s & 5ms \\
30000& 7.68MB &94s & 16ms \\
50000& 12.8MB & 149s & 26ms \\
\hline
\end{tabular}
\caption{Client computation costs for various vector dimensions.} \label{tab:client-comp-cost}
\vspace{-0.5cm}
\normalsize
\end{table}
The memory cost consists of the encrypted vector and the plain vectors in the perturbation pool. According to the algorithm, after each iteration, the resultant $b_i$ will be added to the perturbation pool. However, since the number of iterations is normally small, the pool requires only a limited amount of memory. For example, with 10,000 dimensions and a pool size of 10, 10 iterations need only about 2M bytes of memory to hold the plain long integer vectors.

\subsection{Cloud-side Costs}
The cloud side has major costs in storing the encrypted data and computing the secure matrix-vector multiplication. Note that encrypted data has significantly larger size. It is impractical for an authorized user to download and process the data locally. We show some real numbers in Table \ref{tab:cloud-storage-cost} to give a more concrete idea of the problem scale.

\begin{table}[tbh]
\scriptsize
\centering
\begin{tabular}{|c|c|c|}
\hline
Matrix dimension &Unencrypted (GB) & Encrypted (GB) \\
\hline
10000& 0.8 & 25.8\\
30000&7.2 & 232.0\\
50000&20.0& 645.0\\
\hline
\end{tabular}
\caption{Storage cost for matrices encrypted with a 1024-bit key.} \label{tab:cloud-storage-cost}
\vspace{-0.5cm}
\normalsize
\end{table}
 
Clearly, data in such scales cannot be processed using traditional methods. Instead, we have to fully exploit the parallel processing power in the cloud. The cloud side computation consists of the secure matrix-vector multiplication. The matrix is stored in the form of row vector sets and split into data blocks of 64MB in the Hadoop file system. The MapReduce framework assigns each block to a Map. This allows sets of vector-vector dot products to occur in parallel. Because the expensive operations happen in the Map phase, the number of Maps determines the overall performance. 

\begin{figure}[tbh]
\centering
\begin{minipage}{.9\linewidth}
\centering
\includegraphics[width=\linewidth]{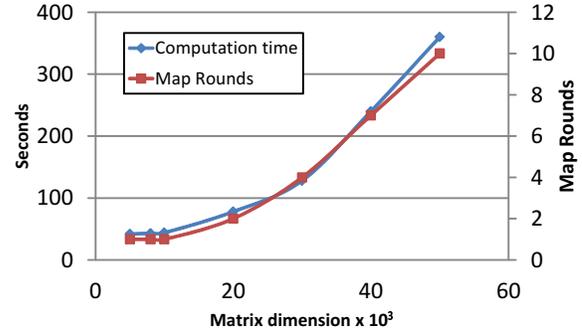}
\vspace{-0.75cm}
\caption{MapReduce costs of one encrypted matrix-vector computation.  Processing cost is approximately determined by the number of Map rounds.}\label{fig:mr-cost-combine}
\end{minipage}
\vspace{-0.5cm}
\end{figure}

Figure \ref{fig:mr-cost-combine} shows the cost increase trend with different sizes of encrypted matrices. With the number of dimensions less than 10,000, the Map phase can be done in one round with the in-house Hadoop cluster. With more dimensions, more Map rounds are needed and the total time cost is proportional to the number of Map rounds. On average each Map round costs about 30-40 seconds. This shows great scalability. As we increase the size of the cluster (with more Map slots) to maintain one Map round, the overall cost will stay constant (around 40 seconds as shown in the figure). 

\begin{figure}[tbh]
\centering
\begin{minipage}{.9\linewidth}
\centering
\includegraphics[width=\linewidth]{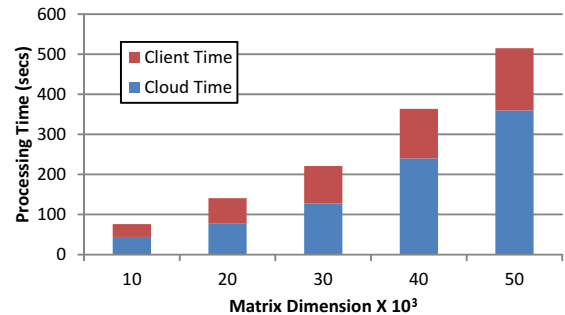}
\vspace{-0.75cm}
\caption{ A summary of the cloud-side and client-side costs shows low client-side cost and the cloud-side efficiency due to the optimized MapReduce implementation. \newline }
\label{fig:cost-summary}
\end{minipage}
\end{figure}
\vspace{-0.5cm}

The overall cost distribution over the cloud and the client side for one iteration using the in-house cluster is shown in Figure \ref{fig:cost-summary}. With increased dimensions, the cloud-side cost dominates the overall cost.

\section{Related Work}\label{sec:related}
A framework for securely outsourcing general computations is given in \cite{gennaro10}.  However, this framework is based on Gentry's fully homomorphic encryption scheme \cite{gentry09} rendering it impractical due to the high computational costs and ciphertext sizes.  A simple test with the Scarab FHE library (hcrypt.com/scarab-library)  yielded ciphertext sizes more than ten times those generated using Paillier \cite{paillier99}.  Very recent work by Naehrig et al. \cite{naehrig11} propose a secure outsourcing solution for problems which only require the encryption scheme to be ``somewhat'' homomorphic (SHE).  They use the SHE scheme of \cite{bv11a} which provides reasonably efficient computational performance but still suffers from large ciphertexts.

Atallah et al. \cite{atallah10} use a more directed approach and present secure outsourcing solutions that are specific to large scale systems of linear equations and matrix multiplication applications.  These solutions fall short in that they leak private information, depend on multiple non-colluding servers, and require a large communication overhead, respectively.  Wang et al. \cite{wang11} use an iterative approach for solving linear equations via client-cloud collaboration.  However, their approach has several  weaknesses.  First, their approach requires that the entire unencrypted matrix be present at the client side.  Secondly, the client side must perform a problem transformation step with a computation cost of $O(n^2)$.  These weaknesses render the approach impractical for big matrices and do not fully utilize the benefit of the cloud.  

Secure multiparty computation (SMC) solutions \cite{nissim06, cramer01, kiltz07} currently exist for various linear algebra problems.  In general, the SMC solutions do not translate well to the secure outsourcing scenario. Each of the SMC parties typically holds a share of the data in plaintext form and conducts computation with it. In our approach, the cloud holds the encrypted data. In addition, the SMC approaches normally have high communication overhead between the parties, which is not desired in cloud-based computing.


\section{Conclusion}\label{sec:conclusion}
In this paper we present an iterative processing approach for finding top-k eigenvectors from encrypted data in the cloud. The security of this approach is implemented with the Paillier encryption system and an efficient random vector perturbation. It is carefully designed so that the client-side computation cost is minimized. We also develop a MapReduce program to efficiently process the cloud-side encrypted matrix-vector multiplication. The experimental results demonstrate the storage and computational costs needed to setup and use the proposed approach and show the scalability of the design. 

In future work, we will optimize our algorithm for sparse matrices and formulate techniques to detect dishonest or lazy service providers.

\bibliographystyle{IEEEtrans}
\bibliography{paper,paper_jim} 

\begin{thebibliography}{10}
\providecommand{\url}[1]{#1}
\csname url@samestyle\endcsname
\providecommand{\newblock}{\relax}
\providecommand{\bibinfo}[2]{#2}
\providecommand{\BIBentrySTDinterwordspacing}{\spaceskip=0pt\relax}
\providecommand{\BIBentryALTinterwordstretchfactor}{4}
\providecommand{\BIBentryALTinterwordspacing}{\spaceskip=\fontdimen2\font plus
\BIBentryALTinterwordstretchfactor\fontdimen3\font minus
  \fontdimen4\font\relax}
\providecommand{\BIBforeignlanguage}[2]{{%
\expandafter\ifx\csname l@#1\endcsname\relax
\typeout{** WARNING: IEEEtranS.bst: No hyphenation pattern has been}%
\typeout{** loaded for the language `#1'. Using the pattern for}%
\typeout{** the default language instead.}%
\else
\language=\csname l@#1\endcsname
\fi
#2}}
\providecommand{\BIBdecl}{\relax}
\BIBdecl

\bibitem{arnoldi51}
W.~E. Arnoldi, ``The principle of minimized iterations in the solution of the
  matrix eigenvalue problem,'' \emph{Quarterly of Applied Mathematics}, vol.~9,
  p. 17–29, 1951.

\bibitem{atallah10}
M.~J. Atallah and K.~B. Frikken, ``Securely outsourcing linear algebra
  computations,'' in \emph{Proceedings of the 5th ACM Symposium on Information,
  Computer and Communications Security}, 2010, pp. 48--59.

\bibitem{bv11a}
Z.~Brakerski and V.~Vaikuntanathan, ``Efficient fully homomorphic encryption
  from (standard) lwe,'' in \emph{Proceedings of the 2011 IEEE 52nd Annual
  Symposium on Foundations of Computer Science}.\hskip 1em plus 0.5em minus
  0.4em\relax ACM Press, 2011, pp. 97--106.

\bibitem{brin98}
S.~Brin and L.~Page, ``The anatomy of a large-scale hypertextual web search
  engine,'' in \emph{{International Conference on World Wide Web}}, 1998.

\bibitem{cramer01}
R.~Cramer and I.~Damg{\aa}rd, ``Secure distributed linear algebra in a constant
  number of rounds,'' in \emph{Proceedings of CRYPTO conference}.\hskip 1em
  plus 0.5em minus 0.4em\relax London, UK: Springer-Verlag, 2001, pp. 119--136.

\bibitem{cullum85}
J.~K. Cullum and R.~A. Willoughby, \emph{Lanczos Algorithms for Large Symmetric
  Eigenvalue Computations}.\hskip 1em plus 0.5em minus 0.4em\relax Cambridge
  University Press, 1985.

\bibitem{deerwester90}
S.~C. Deerwester, S.~T. Dumais, T.~K. Landauer, G.~W. Furnas, and R.~A.
  Harshman, ``Indexing by latent semantic analysis,'' \emph{JASIS}, vol.~41,
  no.~6, pp. 391--407, 1990.

\bibitem{gennaro10}
R.~Gennaro, C.~Gentry, and B.~Parno, ``Non-interactive verifiable computing:
  outsourcing computation to untrusted workers,'' in \emph{Proceedings of
  CRYPTO Conference}.\hskip 1em plus 0.5em minus 0.4em\relax Berlin,
  Heidelberg: Springer-Verlag, 2010, pp. 465--482.

\bibitem{gentry09}
C.~Gentry, ``Fully homomorphic encryption using ideal lattices,'' in
  \emph{{Annual ACM Symposium on Theory of Computing}}.\hskip 1em plus 0.5em
  minus 0.4em\relax New York, NY, USA: ACM, 2009, pp. 169--178.

\bibitem{goodchild07}
M.~F. Goodchild, ``Citizens as sensors: the world of volunteered geography,''
  \emph{Geojournal}, vol.~69, no.~4, pp. 211--221, 2007.

\bibitem{jolliffe86}
I.~T. Jolliffe, \emph{Principal Component Analysis}.\hskip 1em plus 0.5em minus
  0.4em\relax Springer, 1986.

\bibitem{katz07}
J.~Katz and Y.~Lindell, \emph{Introduction to Modern Cryptography}.\hskip 1em
  plus 0.5em minus 0.4em\relax Chapman and Hall/CRC, 2007.

\bibitem{kiltz07}
E.~Kiltz, P.~Mohassel, E.~Weinreb, and M.~Franklin, ``Secure linear algebra
  using linearly recurrent sequences,'' \emph{Theory of Cryptography}, pp.
  291--310, 2007.

\bibitem{lenstra01}
A.~K. Lenstra and E.~R. Verheul, ``Selecting cryptographic key sizes,''
  \emph{Journal of Cryptology}, 2001.

\bibitem{naehrig11}
M.~Naehrig, K.~Lauter, and V.~Vaikuntanathan, ``Can homomorphic encryption be
  practical?'' in \emph{Proceedings of cloud computing security
  workshop}.\hskip 1em plus 0.5em minus 0.4em\relax New York, NY, USA: ACM,
  2011, pp. 113--124.

\bibitem{nissim06}
K.~Nissim and E.~Weinreb, ``Communication efficient secure linear algebra,'' in
  \emph{In the Third Theory of Cryptography Conference – TCC 2006}, 2006.

\bibitem{paillier99}
P.~Paillier, ``Public-key cryptosystems based on composite degree residuosity
  classes,'' in \emph{EUROCRYPT}.\hskip 1em plus 0.5em minus 0.4em\relax
  Springer-Verlag, 1999, pp. 223--238.

\bibitem{regev05}
O.~Regev, ``On lattices, learning with errors, random linear codes, and
  cryptography,'' in \emph{Proceedings of the thirty-seventh annual ACM
  symposium on Theory of computing}.\hskip 1em plus 0.5em minus 0.4em\relax New
  York, NY, USA: ACM, 2005, pp. 84--93.

\bibitem{saad11}
Y.~Saad, \emph{{Numerical Methods for Large Eigenvalue Problems}}.\hskip 1em
  plus 0.5em minus 0.4em\relax New York, NY: SIAM, 2011.

\bibitem{wang11}
C.~Wang, K.~Ren, J.~Wang, and K.~M.~R. Urs, ``Harnessing the cloud for securely
  solving large-scale systems of linear equations,'' in \emph{Proceedings of
  ICDCS}, Washington, DC, USA, 2011, pp. 549--558.

\end{thebibliography}

\end{document}